\documentclass[conference,a4paper]{IEEEtran}

\usepackage{amsmath}
\usepackage{amssymb}
\usepackage{amsthm}
\usepackage[table]{xcolor}

\IEEEoverridecommandlockouts

\newtheorem{definition}{Definition}

\newtheorem{example}[definition]{Example}

\newtheorem{lemma}[definition]{Lemma}
\newtheorem{theorem}[definition]{Theorem}
\newtheorem{corollary}[definition]{Corollary}
\newtheorem{remark}[definition]{Remark}

\newcommand{\rk}{\mathrm{rank}}
\newcommand{\Gr}{\mathcal{G}_q(k,n)}

\newcommand{\field}[1]{\mathbb{#1}}
\newcommand{\C}{\field{C}}
\newcommand{\F}{\field{F}}

\newcommand{\cF}{{\cal F}}

\newcommand{\cA}{{\cal A}}

\newcommand{\cC}{{\cal C}}
\newcommand{\cG}{{\cal G}}

\newcommand{\deff}{\mbox{$\stackrel{\rm def}{=}$}}

\newcommand{\Gauss}[2]{
\left[\begin{array}
{c}#1\\#2\end{array}\right]_{q}
}

\begin{document}

\sloppy

\title{ New Lower Bounds for Constant Dimension Codes}

\author{
  \IEEEauthorblockN{Natalia Silberstein}
  \IEEEauthorblockA{Dep. of Electrical \& Computer Eng.\\
  University of Texas at Austin\\
   Austin, TX, USA\\
   Email: natalys@austin.utexas.edu
    }
  \and
  \IEEEauthorblockN{Anna-Lena Trautmann}
  \IEEEauthorblockA{Inst. of Mathematics\\
    University of Zurich\\
    Zurich, Switzerland\\
    Email: anna-lena.trautmann@math.uzh.ch}
}



\maketitle

\begin{abstract}
  This paper provides new constructive lower bounds for constant dimension codes, using different techniques such as Ferrers diagram rank metric codes and pending  blocks. Constructions for two families of parameters of constant dimension codes are presented. The examples of codes obtained by these constructions are the  largest known constant dimension codes for the given parameters.
\end{abstract}

\section{Introduction}
\label{sec:intro}

Let $\F_q$ be the finite field of size $q$. Given two integers $k,n$, such that $0\leq k\leq n$, the set of all $k$-dimensional subspaces of $\F_q^n$ forms the Grassmannian over $\F_q$, denoted by $\Gr$. It is well known that the cardinality of the Grassmannian is given by the $q$-\emph{ary Gaussian coefficient}
\[\Gauss{n}{k} \deff\ |\Gr|=\prod_{i=0}^{k-1}\frac{q^{n-i}-1}{q^{k-i}-1}.
\]
The Grassmannian space is
a metric space, where the \emph{subspace distance} between any two
subspaces $X$ and $Y$ in $\Gr$, is given by
\begin{equation}
\label{def_subspace_distance}
d_S (X,\!Y) \,\ \deff\ \dim X + \dim Y -2 \dim\bigl( X\, {\cap}Y\bigr).
\end{equation}

We say that $\C\subseteq \Gr$ is an $(n,M,d,k)_q$ \emph{code in
the Grassmannian}, or \emph{constant-dimension code}, if $M =
|\C|$ and $d_S (X,\!Y) \ge d$ for all distinct elements $X,\!Y \in
\C$. Note, that the minimum distance $d$ of $\C$ is always even.
$A_q(n,d,k)$ will denote the maximum size of an $(n,M,d,k)_q$ code.

Constant dimension codes  have drawn a significant attention in the last five years due to
the work by Koetter and Kschischang~\cite{KK}, where they
presented an application of such codes for error-correction in
random network coding. Constructions and bounds for constant dimension codes were
given in~\cite{BoGa09,EtSi09,EtSi12,EV08,ga,GaYa10,KoKu08,MGR08,Ska10,TrRo10}.

In this paper we
focus on constructions of large constant dimension codes.  In particular, we generalize the idea of construction of codes in the Grassmannian from~\cite{EtSi09,EtSi12,TrRo10} and obtain  new lower bounds on $A_q(n,d,k)$. In Section~\ref{sec:preliminaries} we introduce the necessary definitions and present two known constructions which will be the starting point to our new constructions. In Section~\ref{sec:pendingBlocks} we introduce the notation of pending blocks. In Sections~\ref{sec:constructionsk=4} and~\ref{sec:recursion} we present our new constructions. It appears that the codes obtained by these constructions are the largest known constant dimension codes for the given parameters.
%

\section{Preliminaries}
\label{sec:preliminaries}
 In this section we briefly provide the definitions and previous results used in our constructions. More details can be found in~\cite{EtSi09,EtSi12,TrRo10}.

Let $X$ be a $k$-dimensional subspace of $\F_q^n$. We represent $X$ by the matrix $\mbox{RE}(X)$ in reduced
row echelon form, such that the rows of $\mbox{RE}(X)$ form the basis of $X$. The \emph{identifying
vector} of $X$, denoted by $v(X)$ is the binary vector
of length $n$ and weight $k$, where the  $k$ \emph{ones} of $v(X)$ are exactly in the
positions where $\mbox{RE}(X)$ has the leading coefficients (the
pivots).

The {\it Ferrers tableaux form} of a subspace $X$, denoted by
$\cF(X)$, is obtained from $\mbox{RE}(X)$ first by removing from
each row of $\mbox{RE}(X)$ the {\it zeroes} to the left of the
leading coefficient; and after that removing the columns which
contain the leading coefficients. All the remaining entries are
shifted to the right. The \emph{Ferrers diagram} of $X$, denoted
by $\cF_X$, is obtained from $\cF(X)$ by replacing the entries of
$\cF(X)$ with dots.

 Given $\cF(X)$, the unique corresponding
subspace $X\in \Gr$ can be easily found. Also given $v(X)$, the unique corresponding $\cF_X$ can be found. When we fill the dots of a Ferrers diagram by elements of $\F_q$ we obtain a $\cF(X)$ for some $X\in \Gr$.

\begin{example}  Let $X$ be the subspace in $\mathcal G_2(3,7)$ with
the  following  generator matrix in reduced row echelon form:

\begin{footnotesize}
$$\rm{RE}(X)=\left( \begin{array}{ccccccc}
\bf{1} & \color{red}0 & 0 & 0 & \color{red} 1  & \color{red} 1& \color{red} 0\\
0 & 0 & \bf{1} & 0 & \color{red} 1 & \color{red}0 & \color{red}1 \\
0 & 0 & 0 &  \bf{1} &\color{red} 0& \color{red} 1 & \color{red} 1
\end{array}
\right) ~.$$
\end{footnotesize}
Its identifying vector is $v(X)=1011000$, and its Ferrers tableaux
form and Ferrers diagram are given by

\begin{footnotesize}
$$\begin{array}{cccc}
0 & 1 & 1 & 0 \\
&1 & 0 & 1  \\
&0 & 1 & 1
\end{array},~~~
\;
\begin{array}{cccc}
 \bullet & \bullet & \bullet & \bullet \\
  & \bullet & \bullet & \bullet   \\
  & \bullet & \bullet & \bullet  \\
\end{array},$$
\end{footnotesize}
respectively.
\end{example}

In the following we will consider Ferrers diagrams rank-metric codes which are closely related to constant dimension codes.
For two $m \times \ell$ matrices $A$ and $B$ over $\F_q$ the {\it
rank distance}, $d_R(A,B)$, is defined by
$
d_R (A,B) ~ \deff ~ \text{rank}(A-B)~.
$

Let $\cF$ be a Ferrers diagram with $m$ dots in the rightmost
column and $\ell$ dots in the top row. A code $\cC_{\cF}$ is an
$[\cF,\rho,\delta]$ {\it Ferrers diagram rank-metric (FDRM) code}  if
all codewords of $\cC_{\cF}$ are $m\times \ell$ matrices in which
all entries not in $\cF$ are {\it zeroes}, they form a linear subspace of dimension $\rho$ of
$\F_q^{m \times \ell}$, and for any two distinct codewords $A$
and $B$, $d_R (A,B) \geq \delta$.
If $\cF$ is a rectangular $m\times \ell$ diagram with $m\cdot \ell$ dots then the FDRM code is a classical rank-metric code~\cite{Gab85,Rot91}.
The following theorem provides an upper bound on the cardinality of $\cC_{\cF}$.

\begin{theorem}~\cite{EtSi09}\label{thm1}
Let $\cF$ be a Ferrers diagram and $\cC_{\cF}$ the corresponding FDRM code.
Then
$|\cC_{\cF}| \leq q^{\min_i\{w_i\}}$,
where $w_i$ is the number of dots in $\cF$ which are not contained in the first
$i$ rows and the rightmost $\delta-1-i$ columns ($0\leq i\leq \delta-1$).
\end{theorem}

A code which attains the bound of Theorem~\ref{thm1} is
called a \emph{Ferrers diagram maximum rank distance (FDMRD) code}.

\begin{remark}Maximum rank distance (MRD) codes are a class of $[\cF,\ell(m-\delta+1),\delta]$ FDRM codes, $\ell\geq m$, with a full $m\times \ell$ diagram $\cF$, which attain the bound of Theorem~\ref{thm1}~\cite{Gab85,Rot91}.
\end{remark}

It was proved in~\cite{EtSi09} that for general diagrams the bound of Theorem~\ref{thm1} is always attained for $\delta=1,2$.
Some special cases, when this bound is attained for $\delta>2$, can be found in~\cite{EtSi09}.

For a codeword $A \in \cC_{\cF} \subset \F_q^{k \times (n-k)}$
let $A_{\cF}$ denote the part of
$A$ related to the entries of $\cF$ in $A$. Given a
 FDMRD code $\cC_{\cF}$, a lifted FDMRD code
$\C_{\cF}$ is defined as follows:
$$\C_{\cF} = \{X\in \Gr :
\cF(X)=A_{\cF},~ A \in \cC_{\cF} \}.
$$

This definition is the generalization of the definition of a lifted
MRD code~\cite{SKK}. Note, that all the codewords of a lifted MRD code have the same identifying vector of the type $(\underset{k}{\underbrace{11...1}}\underset{n-k}{\underbrace{000...00}})$.
The following lemma~\cite{EtSi09} is the generalization of the result
given in~\cite{SKK}.

\begin{lemma}
\label{lem:dist_lift} If $\cC_{\cF} \subset \F_q^{k \times (n-k)}$
is an $[ \cF , \rho , \delta ]$ Ferrers diagram
rank-metric code, then its lifted code $\C_{\cF}$ is an $(n, q^\rho ,
2\delta , k)_q$ constant dimension code.
\end{lemma}

\subsection {The multilevel construction and pending dots construction}

It was proved in~\cite{EtSi09} that for
$X,Y\in\Gr$ we have $d_S(X,Y)\geq d_H(v(X),v(Y))$, where $d_H$
denotes the Hamming distance; and  if $v(X)=v(Y)$ then
$d_S(X,Y)=2d_R(\mbox{RE}(X),\mbox{RE}(Y))$.
 The multilevel construction~\cite{EtSi09} of constant dimension code is based on these properties of $d_S$.

\textbf{ Multilevel construction.}
 First,  a binary constant weight code of
length $n$, weight $k$, and Hamming distance $2 \delta$ is
chosen to be the set of the identifying vectors
for~$\C$. Then, for each identifying vector a corresponding lifted
FDMRD code with minimum rank distance $\delta$
is constructed. The union of these lifted FDMRD
codes is an $(n, M, 2\delta, k)_q$ code.

In the construction provided in~\cite{EtSi12}, for $k=3$ and $\delta=2$, in the stage of choosing identifying vectors for a code $\C$, the vectors of (Hamming) distance $2\delta-2=2$ are allowed, by using a method based on pending dots in a Ferrers
diagram~\cite{TrRo10}.


The \emph{pending dots} of a Ferrers diagram $\cF$ are the leftmost dots in the
first row of $\cF$ whose removal has no impact on the size of the
corresponding Ferrers diagram rank-metric code. The following lemma follows from~\cite{TrRo10}.

\begin{lemma}\cite{TrRo10}\label{lm:pending dots}
Let $X$ and $Y$  be two subspaces in $\Gr$ with
$d_H(v(X),v(Y))=2\delta-2$, such that the leftmost \emph{one} of
$v(X)$ is in the same position as the leftmost  \emph{one} of
$v(Y)$.  Let $P_X$ and $P_Y$ be the sets of pending dots of $X$ and $Y$, respectively.
If $P_X\cap P_Y\neq \varnothing$
and the entries in $P_X\cap P_Y$ (of their Ferrers tableaux forms) are
assigned with different values in at least one position, then
$d_S(X,Y)\geq 2\delta.$
\end{lemma}

\begin{example}
Let $X$ and $Y$  be subspaces in $\mathcal{G}_q(3,6)$ which are given by
the following generator matrices:
\begin{footnotesize}
$$\left(\begin{array}{cccccc}
1 &\textcircled{\raisebox{-0.9pt}{0}} &0  & v_1 & v_2 &  0\\
0 & 0 & 1  & v_3 & v_4 &  0 \\
0 & 0 & 0 &  0 &  0 & 1
\end{array}\right),\:
\left(\begin{array}{cccccc}
1 &\textcircled{\raisebox{-0.9pt}{1}} &u_1 & 0 &  u_2  & 0\\
0 & 0 & 0  & 1  & u_3 & 0  \\
0 & 0 & 0 & 0 &  0  & 1
\end{array}\right)
$$
\end{footnotesize}
where $v_i ,u_i\in \F_q$, and the pending dots are emphasized by
circles. Their identifying vectors are $v(X)=101001$ and
$v(Y)=100101$. Clearly, $d_H(v(X), v(Y))=2$, while $d_S(X,Y)\geq 4$.
\end{example}

The following lemma which  follows from a one-factorization and near-one-factorization of  a complete graph~\cite{vLWi92} will be used in our constructions.

\begin{lemma}\label{lm:1-factorization}
Let $D$ be the set of all binary vectors  of length~$m$ and weight
$2$.
\begin{itemize}
\item If $m$ is even, $D$ can be partitioned into $m-1$ classes,
each of $\frac{m}{2}$ vectors with pairwise disjoint positions of \emph{ones};
\item If $m$ is odd, $D$ can be partitioned into $m$ classes,
each of $\frac{m-1}{2}$ vectors with pairwise disjoint positions of \emph{ones}.
\end{itemize}
\end{lemma}

The following construction  for $k=3$ and $d=4$ based on pending dots~\cite{EtSi12} will be used as a base step of our recursive construction proposed in the sequel.

\textbf{Construction 0.}
 Let $n\geq 8$ and $q^2+q+1\geq \ell$, where $\ell=n-4$ for odd $n$ and $\ell=n-3$ for even $n$.
 In addition to the lifted MRD code (which has the identifying vector $v_0=(11100\ldots0)$), the final code $\C$ will contain the codewords with identifying vectors of the form $(x||y)$, where the prefix $x\in \F_2^3$ is of weight $1$ and the suffix $y\in \F_2^{n-3}$ is of weight~$2$. By Lemma~\ref{lm:1-factorization}, we partition the set of suffixes into $\ell$  classes $P_1,P_2,\ldots,P_{\ell}$ and define the following three sets:
\[\cA_1=\{(001||y):y\in P_1\},\]
\[\cA_2=\{(010||y):y\in P_i, 2\leq i\leq \min\{q+1,\ell\}\},\]
\[\cA_3= \left\{\begin{array}{cc}
                \{(100||y):y\in P_i,~ q+2\leq i\leq \ell\} & \textmd{if }\ell>q+1 \\
                \varnothing & \textmd{if }\ell\leq q+1 \\
              \end{array}\right..
\]
The idea  is that we use the same prefix for the suffixes of Hamming distance $4$ (from the same class), and when we use the same prefix for two different classes $P_i, P_j$, we use the different values of  Ferrers tableaux forms  in the pending dots.  Then, the corresponding lifted FDMRD codes of distance $4$ are constructed, and their union with the lifted MRD code forms the final code $\C$ of size $q^{2(n-3)}+\Gauss{n-3}{2}$.

In the following sections we will generalize this construction and obtain codes for any $k\geq4$ with $d=4$ and  with $d=2(k-1)$.

\section{Pending Blocks}
\label{sec:pendingBlocks}

To present the new constructions for constant dimension codes,
we first need to extend the definition of pending dots of~\cite{TrRo10} to a two-dimensional setting.

\begin{definition}
Let $\cF$ be a Ferrers diagram with $m$ dots in the rightmost column and $\ell$ dots in the top row. We say that the $\ell_1<\ell$ leftmost columns of $\cF$ form a \emph{pending block} (of size $\ell_1$) if the upper bound on the size of FDMRD code $\cC_{\cF}$  from Theorem~\ref{thm1} is equal to the upper bound on the size of $\cC_{\cF}$ without the $\ell_1$ leftmost columns.
\end{definition}

\begin{example}\label{ex:pending_block}
Consider the following Ferrers diagrams:
\begin{footnotesize}
\[ \cF_1=\begin{array}{cccccc} \bullet & \bullet &\bullet&\bullet&\bullet \\
 &\bullet& \bullet & \bullet &\bullet\\
& &\bullet &\bullet&\bullet
\end{array}\quad, \quad
\cF_2=\begin{array}{cccccc}\bullet&\bullet&\bullet \\
  \bullet & \bullet &\bullet\\
\bullet &\bullet&\bullet
\end{array} .\]
\end{footnotesize}

For $\delta=3$ by Theorem~\ref{thm1} both codes $\cC_{\cF_1}$ and $\cC_{\cF_2}$ have $|\cC_{\cF_i}|\leq q^3$, $i=1,2$.
The diagram $\cF_1$ has the pending block $ \begin{array}{cccccc} \bullet & \bullet\\&\bullet
\end{array}$ and the diagram $\cF_2$ has no pending block.
\end{example}

\begin{definition}
Let $\cF$ be a Ferrers diagram with $m$ dots in the rightmost column and $\ell$ dots in the top row, and let $\ell_1< \ell$, and $m_1< m$. If the $(m_1+1)$st row of $\cF$ has less dots than the $m_1$th row of $\cF$, then the $\ell$ leftmost columns of $\cF$ are called a \emph{quasi-pending block} (of size $m_1\times\ell_1$).
\end{definition}

Note, that a pending block is also a quasi-pending block.


\begin{theorem}\label{thm4}
Let $X,Y\in \Gr$, such that $\rm RE(X)$ and $\rm RE(Y)$ have a quasi-pending block of size $m_{1}\times \ell_{1}$ in the same position and $d_{H}(v(X),v(Y))= d$.   Denote the submatrices of $\cF(X)$ and $\cF(Y)$ corresponding to the quasi-pending blocks by $B_X$ and $B_Y$, respectively.
Then
$d_{S}(X,Y) \geq d+2\rk(B_X- B_Y) $.
\end{theorem}



\begin{proof}
Since the quasi-pending blocks are in the same position, it has to hold the first $h$ pivots of $\rm RE(X)$ and $\rm RE(Y)$ are in the same columns. To compute the rank of $\left[ \begin{array}{cc}\rm RE(X)\\\rm RE(Y)\end{array} \right] $ we permute the columns such that the $h$ first pivot columns are to the very left, then the columns of the pending block, then the other pivot columns and then the rest (WLOG in the following figure we assume that the $h+1$st pivots are also in the same column):

\begin{footnotesize}
\[\rk\left[ \begin{array}{cccccccccccccc}
1&\dots&0 &\cellcolor{gray}  &\cellcolor{gray}&\cellcolor{gray} &\cellcolor{gray}&\cellcolor{gray} &0&\dots\\
\vdots&\ddots &\vdots&&&\cellcolor{gray}\ddots&\cellcolor{gray}B_X&\cellcolor{gray}\vdots&\vdots&\vdots\\
0&\dots&1&0&\dots&0&\cellcolor{gray} &\cellcolor{gray} &0&\dots\\
0&\dots&0&0&\hdots&0&\hdots&0&1&\hdots\\
\vdots&&&&&&&&\vdots\\
\hline
1&\dots&0 &\cellcolor{gray}  &\cellcolor{gray}&\cellcolor{gray} &\cellcolor{gray}&\cellcolor{gray} &0&\dots\\
\vdots&\ddots &\vdots&&&\cellcolor{gray}\ddots&\cellcolor{gray}B_Y&\cellcolor{gray}\vdots&\vdots&\vdots\\
0&\dots&1&0&\dots&0&\cellcolor{gray} &\cellcolor{gray} &0&\dots\\
0&\dots&0&0&\hdots&0&\hdots&0&1&\hdots\\
\vdots&&&&&&&&\vdots
\end{array}\right]\]
\end{footnotesize}
Now we subtract the lower half from the upper one and get
\begin{footnotesize}
\[=\rk\left[ \begin{array}{cccccccccccccc}
1&\dots&0 &\cellcolor{gray}  &\cellcolor{gray}&\cellcolor{gray} &\cellcolor{gray}&\cellcolor{gray} &0&\dots\\
\vdots&\ddots &\vdots&&&\cellcolor{gray}\ddots&\cellcolor{gray}B_X&\cellcolor{gray}\vdots&\vdots&\vdots\\
0&\dots&1&0&\dots&0&\cellcolor{gray} &\cellcolor{gray} &0&\dots\\
0&\dots&0&0&\hdots&0&\hdots&0&1&\hdots\\
\vdots&&&&&&&&\vdots\\
\hline
0&\dots&0 &\cellcolor{gray}  &\cellcolor{gray}&\cellcolor{gray} &\cellcolor{gray}&\cellcolor{gray} &0&\dots\\
\vdots&\ddots &\vdots&&&\cellcolor{gray}\ddots&\cellcolor{gray}B_X-B_Y&\cellcolor{gray}\vdots&\vdots&\vdots\\
0&\dots&0&0&\dots&0&\cellcolor{gray} &\cellcolor{gray} &0&\dots\\
0&\dots&0&0&\hdots&0&\hdots&0&0&\hdots\\
\vdots&&&&&&&&\vdots
\end{array}\right]\]
\end{footnotesize}
The additional pivots of $\rm RE(X)$ and $\rm RE(Y)$ (to the right in the above representation) that were in different columns in the beginning are still in different columns, hence it follows that
$$\rk\left[ \begin{array}{cc}\rm RE(X)\\\rm RE(Y)\end{array} \right] \geq  k+\frac{1}{2}d_{H}(v(X),v(Y))+\rk(B_X-B_Y) ,$$
which implies the statement with the formula
\[d_{S}(X,Y)=2\rk\left[ \begin{array}{cc}\rm RE(X)\\\rm RE(Y)\end{array} \right] -2k .\]
\end{proof}

This theorem implies that for the construction of an $(n,M,2\delta,k)$-code, by filling the (quasi-)pending blocks with a suitable Ferrers diagram rank metric code, one can choose a set of identifying vectors with lower minimum Hamming distance than $\delta$.

\section{  Constructions for  $(n,M,4,k)_q$ Codes}
\label{sec:constructionsk=4}

In this section we present a construction based on quasi-pending blocks  for $(n,M,4,k)_q$ codes with $k\geq 4$ and $n\geq 2k+2$.
This construction will then give rise to new lower bounds on the size of constant dimension codes with this minimum distance.
First we need the following results.

\begin{lemma}
\label{trm:k-2 ones}
Let $n\geq 2k+2$.
 Let $v$ be an identifying vector of length $n$ and weight $k$, such that there are $k-2$ many ones in the first $k$ positions of $v$. Then the Ferrers diagram  arising from $v$ has more or equally many dots in the first row than in the last column,
 and the upper bound for the dimension of a Ferrers diagram code with minimum distance $2$ is the number of dots that are not in the first row.
\end{lemma}

\begin{proof}
 Because of the distribution of the ones, it holds that the number of dots in the first row of the Ferrers diagram is
\[n- k -2+ i ,\quad i\in\{0,1,2\}\]
and the number of dots in the last column of the Ferrers diagram is
\[k-2+j ,\quad j\in\{0,1,2\} .\]
Since we assume that $n\geq 2k+2$, the number of dots in the first row is always greater or equal to the number of dots in the last column. Then, the upper bound on the dimension directly follows from Theorem~\ref{thm1}.
\end{proof}


%

\begin{lemma}\label{lem:cardblocks}
 The number of all matrices filling the Ferrers diagrams arising from all elements of $\F_q^k$ of weight $k-2$ as identifying vectors is
$\nu:=\sum_{j=0}^{k-2}\sum_{i=j}^{k-2} q^{i+j}  -1   .$
\end{lemma}

\begin{proof}
Assume the first zero is in the $j$-th and the second zero is in the $i$-th position of the identifying vector. Then the corresponding Ferrers diagram has $j-1$ dots in the first column and $i-2$ dots in the second column. I.e. there are
\[\sum_{j=1}^{k-1}  \sum_{i=j+1}^{k} {(j-1)+(i-2)} =\sum_{j=0}^{k-2}\sum_{i=j}^{k-2} {i+j}  \]
dots over all and we can fill each diagram with $i$ dots with $q^i$ many different matrices.
Hence the formula follows, since we have to subtract $1$ for the summand where $i=j=0$.
\end{proof}

%
%
%
%
%
%
%
%
%
%

We can now describe the construction ($k\geq 4, n\geq 2k+2$):

\textbf{Construction Ia.} \label{const1}
First, by Lemma~\ref{lm:1-factorization}, we partition the weight-$2$ vectors of $\F_{2}^{n-k}$ into classes $P_{1},\dots, P_{\ell}$ of size $\frac{\bar \ell}{2}$  (where $\ell=\bar \ell -1=n-k-1$ if $n-k$ even and $\ell=\bar \ell +1=n-k$ if $n-k$ odd) with pairwise disjoint positions of the ones.
\noindent
\begin{enumerate}
 \item We define the following sets of identifying vectors (of weight $k$):
\begin{align*}
\cA_0&=\{(1\dots1|| 0\dots 0)\}\\
\cA_1&=\{(0011\dots 1 || y)  : y \in P_{1}\},\\
\cA_2&=\{(0101\dots 1 || y)  : y \in P_{2},\dots,P_{q+1}\},\\
\vdots \\
\cA_{\binom{k}{2}}&=\{(1\dots110 0 || y )  : y \in P_{\mu},\dots,P_{\nu}\}.
\end{align*}
such that the prefixes in $\cA_{1}, \dots, \cA_{\binom{k}{2}}$ are all vectors of $\F_2^{k}$ of weight $k-2$.
The number of $P_{i}$'s used in each set depends on the size of the quasi-pending block arising in the $k$ leftmost columns of the respective matrices. Thus, $\nu$ is the value from Lemma \ref{lem:cardblocks} and $\mu:=\nu-q^{2(k-2)}$.
\item For each vector $v_j$ in a given $\cA_i$ for $i\in \{{2},\dots,{\binom{k}{2}}\}$ assign a different matrix filling for the quasi-pending block in the $k$ leftmost columns of the respective matrices.
 Fill the  remaining part of the Ferrers diagram with a suitable FDMRD code of the minimum rank distance $2$ and lift the code to obtain $\C_{i,j}$. Define $\C_i=\bigcup_{j=1}^{|\cA_i|}\C_{i,j}$.
\item Take the largest known code $\bar \C \subseteq \mathcal{G}_{q}(k, n-k)$ with minimum distance $4$ and append $k$ zero columns in front of every matrix representation of the codewords.
\item
     The following union of codes form the final code $\C$:
     \[\C=\bigcup_{i=0}^{\binom{k}{2}}\C_i  \cup\bar \C
     \]
     where $\C_{0}$ is the lifted MRD code corresponding to $\cA_{0}$.
\end{enumerate}


\begin{remark}
If $\ell < \nu$, then we use only the sets $\cA_{0},\dots,\cA_{i}$ ($i\leq \binom{k}{2}$) such that all of $P_{1},\dots,P_{\ell}$ are used once.
\end{remark}

\begin{theorem}
If $\ell \leq \nu$, a code $\C \subseteq \Gr$ constructed according to Construction I has minimum distance $4$ and cardinality
{\footnotesize
\[
|\C|=q^{(k-1)(n-k)} + q^{(n-k-2)(k-3)} \Gauss{n-k}{2}
+ A_{q}(n-k, 4, k) .
\]}
\end{theorem}
\begin{proof}
It holds that $|\C_{0}|=q^{(k-1)(n-k)}$ and $|\bar \C| = A_{q}(n-k, 4, k)$.
Because of the assumption on $k$ and $q$ it follows from Lemma \ref{lem:cardblocks} that all the $y_{i} \in \F_{2}^{n-k}$ are used for the identifying vectors, hence a cardinality of $|\cG_q(2,n-k)|$ for the lower two rows.
Moreover, we can fill the second to $(k-2)$-nd row of the Ferrers diagrams with anything in the construction of the FDMRD code, hence $q^{(n-k-2)(k-3)}$ possibilities for these dots.
Together with 4) from Construction I we have the lower bound on the code size.

Let $X,Y \in \C$ be two codewords. If both are from $\bar \C$, the distance is clear. If  $X$ is from $\bar \C$ and $Y$ is not, then $d_{H}(v(X),v(Y))\geq 2(k-2)$. Since $k\geq 4$, it follows that $d_{S}(X,Y)\geq 4$. For the rest we distinguish four different cases:
\begin{enumerate}
 \item If $v(X)=v(Y)$, then the FDMRD code implies the distance.
 \item If $v(X)\neq v(Y)$ and $v(X),v(Y)$ are in the same set $\cA_{i}$ for some $i$, then $d_{H}(v(X),v(Y))\geq 2$ (because of the structure of the $P_{i}$'s). The quasi-pending blocks then imply by Theorem~\ref{thm4} that $d_{S}(X,Y) \geq 2+2=4$.
 \item If $v(X) \in \cA_{0},v(Y)\in \cA_{j}$, where $j>0$, then $d_{H}(v(X),v(Y)\geq 4$. Hence, $d_{S}(X,Y) \geq 4$.
 \item If $v(X) \in \cA_{i},v(Y)\in \cA_{j}$, where $i\neq j$ and $i,j>0$, then $d_{H}(v(X),v(Y)\geq 4$ because the first $k$ coordinates have minimum distance $\geq 2$ and the last $n-k$ coordinates have minimum distance $\geq 2$, since they are in different $P_{i}$'s. Hence, $d_{S}(X,Y) \geq 4$.
\end{enumerate}
\end{proof}

%
%

We can now retrieve a lower bound on the size of constant dimension codes for minimum subspace distance $4$.
\begin{corollary}\label{bound1}
Let $k\geq 4, n\geq 2k+2$ and $\sum_{j=0}^{k-2}\sum_{i=j}^{k-2} q^{i+j}    -1 \geq n-k$ if $n-k$ is  odd (otherwise $\geq n-k-1$). Then
{\footnotesize $$A_{q}(n, 4, k) \geq q^{(k-1)(n-k)} + q^{(n-k-2)(k-3)} \left[\begin{array}{c}n-k\\2  \end{array} \right]_q + A_{q}(n-k, 4, k).$$}
\end{corollary}
This bound is always tighter than the ones given by the Reed-Solomon like construction \cite{KK} and the multicomponent extension of this \cite{ga,tr}.

%


Note, that in the construction we did not use the dots in the
quasi-pending blocks
for the calculation of the size of a FDMRD code. Thus, the bound of Corollary~\ref{bound1} is not tight.
To make it tighter, one can use less pending blocks and larger FDMRD codes, as illustrated in the following construction.
We denote by $P_y$ the class of suffixes which contains the suffix vector $y$ (in the partition of Lemma~\ref{lm:1-factorization}).

\textbf{Construction Ib.} First, in addition to $\cA_0$ of Construction~Ia, we define the following sets of identifying vectors:
$$\bar\cA_1=\{(11...1100||y): y \in P_{1100...00}\},
$$
$$\bar\cA_2=\{(11...1010||y): y \in P_{1010...00}\},
$$
$$\bar\cA_3=\{(11...0110||y): y \in P_{1001...00}\},
$$
$$\bar\cA_4=\{(11...1001||y): y \in P_{0110...00}\}.
$$
All the other identifying vectors are distributed as  in Construction Ia. The steps $2)-4)$ of Construction Ia remain the same.
Then the lower bound on the cardinality becomes
\begin{corollary} If $\sum_{j=0}^{k-2}\sum_{i=j}^{k-2} q^{i+j}    - \sum_{i=4}^{5}q^{2k-i} -2q^{2k-6}\geq n-k$, then
{\footnotesize $$
A_{q}(n, 4, k) \geq q^{(k-1)(n-k)} + q^{(n-k-2)(k-3)} \Gauss{n-k}{2}
$$
$$
+ (q^{2(k-3)}-1)q^{(k-1)(n-k-2)}+(q^{2(k-3)-1}-1)q^{(k-1)(n-k-2)-1}
$$
$$
+2(q^{2(k-4)}-1)q^{(k-1)(n-k-2)-2}+A_{q}(n-k, 4, k).
$$}
\end{corollary}
Note, that one can use this idea on more $\cA_{i}$'s, as long as there are enough pending blocks such that all $P_{i}$'s are used.

Moreover, instead of using all the classes $P_i$ we can use the classes which contribute more codewords more then once with the disjoint prefixes.
We illustrate this idea for a code having $k=4$ and $n=10$.
It appears, that the code obtained by this construction is the largest known code.

\begin{example}
Let $q=2$, $k=4$, $n=10$. We partition the binary vectors of length 6 and weight 2 into the following 5 classes:
$P_1=\{110000,001010,000101\},
 P_2= \{101000, 010001,000110\},
P_3= \{011000, 100100, 000011\},\\
P_4= \{010100, 100010,001001\},
P_5= \{100001, 010010,\\001100\}.
$
We define $\cA_{0}$ as previously and
$$\cA_1=\{(1100||y) : y \in P_1\}, \cA_2=\{(0011||y) : y \in P_1\},
$$
$$\cA_3=\{(0110||y) : y \in P_4\},
\cA_4=\{(1001||y) : y \in P_4\},
$$
{\small
$$\cA_5=\{(1010||y) : y \in P_2\cup P_3\},
\cA_6=\{(0101||y) : y \in P_2\cup P_3\} ,
$$}
where we use the pending dot in $\cA_{5}$ and $\cA_{6}$.
 Note, that we do not use $P_{5}$. Also, the FDMRD codes are now constructed for the whole Ferrers diagrams (without the pending dot), and not only for the last $6$ columns.
We can add $A_{2}(6,4,4)=A_{2}(6,4,2) = (2^{6}-1)/(2^{2}-1)= 21$ codewords corresponding to set 4) in Construction Ia. The size of the final code is $2^{18}+37477$. The largest previously known code was obtained by the multilevel construction and has size $2^{18}+34768$~\cite{EtSi09}.
\end{example}

In the following we discuss a construction of a  new constant dimension code with minimum distance $4$ from a given one.

\begin{theorem}
\label{trm:code extension}
Let $\C$ be an $(n,M,4,k)_q$ constant dimension code. Let $\Delta$ be an integer such that $\Delta\geq k$.
 Then, there exists an $(n'=n+\Delta, M', 4,k)$ code $\C'$  with $M'=Mq^{\Delta(k-1)}$.
\end{theorem}

\begin{proof}
To the generator matrix of each codeword of $\C$ we append a $[k\times \Delta, \Delta(k-1), 2]$-MRD  code in the additional columns. This MRD code has cardinality $q^{\Delta (k-1)}$ by Theorem~\ref{thm1}.
\end{proof}


\begin{example}We take the  $(8,2^{12}+701,4,4)_2$  code $\C$ constructed in~\cite{EtSi12} and apply on it Theorem~\ref{trm:code extension} with $\Delta=4$.
Then the code $|\C'|=2^{24}+701\cdot 2^{12}=2^{24} + 2871296$.
The largest previously known code of size $2^{24}+2290845$ was obtained in~\cite{EtSi09}.
\end{example}


%
%
%
%
%
%


\section{Construction for $(n,M,2(k-1),k)_{q}$ Codes}
\label{sec:recursion}

In this section we provide a recursive construction for  $(n,M,2(k-1),k)_{q}$  codes, which uses the pending dots based construction 
described in Section~\ref{sec:preliminaries} as an initial step. Codes obtained by this construction contain the lifted MRD code. An upper bound on the cardinality of such codes is given in~\cite{EtSi12}.
%
%
The codes obtained by Construction 0 attain this bound for $k=3$.
Our recursive construction provides a new lower bound on the cardinality of such codes for general $k$.
%
%

First, we need the following lemma which is a simple generalization of Lemma~\ref{trm:k-2 ones}.
\begin{lemma}
\label{lm:gen_k-2}
 Let $n-k-2\geq n_1\geq k-2 $ and $v$ be an identifying vector of length $n$ and weight $k$, such that there are $k-2$ many ones in the first $n_1$ positions of $v$. Then the Ferrers diagram arising from $v$ has more or equally many dots in any of the first $k-2$ rows  than in the last column,
 and the upper bound for the dimension of a Ferrers diagram code with minimum distance $k-1$ is the number of dots that are not in the first $k-2$ rows.
\end{lemma}
 \begin{proof}
 Naturally, the last column of the Ferrers diagram has at most $k$ many dots. Since there are $k-2$ many ones in the first $n_{1}$ positions of $v$, it follows that there are $n-n_{1}-2$ zeros in the last $n-n_{1}$ positions of $v$. Thus, there are at least $n-n_{1}-2$ many dots in any but the lower two rows of the Ferrers diagram arising from $v$. Therefore, if $n-n_{1}-2\geq k \iff n-k-2\geq n_{1}$ the Ferrers diagram arising from $v$ has more or equally many dots in any of the first $k-2$ rows  than in the last column. It holds that any column has at most as many dots as the last one.

 From Theorem \ref{thm1} we know that the bound on the dimension of the FDRM code is given by the minimum of dots not contained in the first $i$ rows and last $k-2-i$ columns for $i=0,\dots,k-2$. Since, for the given $i$'s, the previous statement holds, the minimum is attained for $i=k-2$.
 \end{proof}

\begin{remark}\label{rem5}
If a $m\times \ell$-Ferrers diagram has $\delta$  rows with $\ell$ dots each, then the construction of~\cite{EtSi09} provides respective FDMRD codes of minimum distance $\delta+1$
attaining the bound of Theorem~\ref{thm1}.
\end{remark}


\begin{lemma}\label{lemk-3}
For a $m\times \ell$-Ferrers diagram where the $j$th row has at least $x$ more dots than the $(j+1)$th row for $1\leq j \leq m-1$ and the lowest row has $x$ many dots, one can construct a FDMRD code with minimum rank distance $m$ and cardinality $q^{x}$ as follows. For each codeword take a different $w\in \F_{q}^{x}$ and fill the first $x$ dots of every row with this vector, whereas all other dots are filled with zeros.
\end{lemma}

\begin{proof}
The minimum distance follows easily from the fact that the positions of the $w$'s in each row have no column-wise intersection. Since they are all different, any difference of two codewords has a non-zero entry in each row and it is already row-reduced.

The cardinality is clear, hence it remains to show that this attains the bound of Theorem \ref{thm1}. Plugging in $i=k-1$ in Theorem \ref{thm1} we get that the dimension of the code is less than or equal to the number of dots in the last row, which is achieved by this construction.
\end{proof}

\textbf{Construction II.}
Let $s=\sum_{i=3}^{k}i$,
$n\geq s+2+k$ and
$q^2+q+1\geq \ell$, where $\ell=n-s$ for odd $n-s$ (or $\ell= n-s-1$ for even $n-s$).

\emph{ Identifying vectors:}
 In addition to the identifying vector $v_{00}^k=(11\ldots1100\ldots0)$ of the lifted MRD code $\C^k_*$ (of size $q^{2(n-k)}$ and distance $2(k-1)$), the other identifying vectors of the codewords are defined as follows.
First, by Lemma~\ref{lm:1-factorization}, we partition the weight-$2$ vectors of $\F_{2}^{n-s}$ into classes $P_{1},\dots, P_{\ell}$ of size $\frac{\bar \ell}{2}$  (where $\ell=\bar \ell -1=n-s-1$ if $n-s$ even and $\ell=\bar \ell +1=n-s$ if $n-s$ odd) with pairwise disjoint positions of the ones.
 We define the sets of identifying vectors  by a recursion. Let $v_0$ and $\cA_1,\cA_2,\cA_3\subseteq\F_q^{n-s+3}$,  as defined in Construction~0. Then $v_{00}^3=v_0$,
\[\cA_0^3= \emptyset\textmd{, } \cA_i^3=\cA_i\textmd{, } 1\leq i\leq 3.
\]
For $k\geq 4$ we define:
 \[\cA_0^k=\{v^k_{01},\dots,v^k_{0k-3}\},\]
where $v_{0j}^k=(000\; w^k_j \;|| v_{0j-1}^{k-1})$ ($1\leq j\leq k-3$),
 such that the $w_j^k$ are all different weight-$1$ vectors of $\F_2^{k-3}$.
Furthermore we define:
\begin{align*}
\cA_1^k&=\{(0010\dots 00 || z)  : z \in \cA_1^{k-1}\},\\
\cA_2^k&=\{(0100\dots 00 || z)  : z \in \cA_2^{k-1}\},\\
\cA_{3}^k&=\{(1000\dots00 || z )  : z \in \cA_3^{k-1}\},
\end{align*}
such that the prefixes of the vectors in $\cup_{i=0}^{3}\cA_i^k$ are  vectors of $\F_2^{k}$ of weight $1$.
Note, that the suffix $y\in\F_q^{n-s}$ (from Construction~0)
in all the vectors from $\cA_1^k$ belongs to $P_1$, the suffix $y$ in all the vectors from $\cA_2^k$ belongs to $\cup_{i=2}^{\min\{q+1,\ell\}} P_i$, and the suffix $y$ in all the vectors from $\cA_3^k$ belongs to $\cup_{i=q+2}^{\ell} P_i$ (the set $\cA_3^k$ is empty if $\ell\leq q+1$).

\emph{Pending blocks:}
\begin{itemize}
  \item All Ferrers diagrams that correspond to the vectors  in $\cA_1^k$ have a common pending block with $k-3$ rows and 
  $ \sum_{i=3}^{k-j}i$ dots in the 
  $j$th row, for $1\leq j\leq k-3$. We fill each of these pending blocks with a different element of a suitable FDMRD code with minimum rank distance $k-3$ and size $q^3$, according to Lemma \ref{lemk-3}.
Note, that the initial conditions imply that $q^{3}\geq \bar \ell $, i.e. we always have enough fillings for the pending block to use all elements of the given $P_{i}$.
  \item All Ferrers diagrams that correspond to the vectors  in $\cA_2^k$ have a common pending block with $k-2$ rows and $\sum_{i=3}^{k-j}i+1$ dots in the 
$j$th row, $1\leq j\leq k-2$. Every vector which has a suffix $y$ from the same $ P_i$ will have the same value $a_i\in \F_q$ in the first entry in each row of the common pending block, s.t. the vectors with suffixes from the different classes will have different values in these entries.  (This corresponds to a FDMRD code of distance $k-2$ and size $q$.)
  Given the filling of the first entries of every row, all the other entries of the pending blocks are filled by a FDMRD code with minimum distance $k-3$, according to Lemma \ref{lemk-3}.
  \item All Ferrers diagrams that correspond to the vectors  in $\cA_3^k$ have a common pending block with $k-2$ rows and  $\sum_{i=3}^{k-j}i +2$ dots in the $j$th row, ${1\leq j\leq k-2}$. The filling of these pending blocks is analogous to the previous case, but for the suffixes from the different $P_{i}$-classes we fix the first two entries in each row of a pending block. Hence, there are $q^{2}$ different possibilities.

\end{itemize}

\emph{Ferrers tableaux forms:}
On the dots corresponding to the last $n-s-2$ columns of the Ferrers diagrams for each vector $v_j$ in a given $\cA^k_i$, $0\leq i\leq 3$,  we construct a FDMRD code with minimum distance $k-1$ (according to Remark \ref{rem5}) and lift it to obtain $\C_{i,j}^k$. We define $\C_i^k=\bigcup_{j=1}^{|\cA^k_i|}\C_{i,j}^k$.

\emph{Code:}
The final code is defined as
$$\C^k=\bigcup_{i=0}^{3}\C_{i}^k\cup \C^k_*.$$

\begin{theorem}
\label{trm:recursive_parameters}
The code $\C^k$ obtained by Construction II has minimum distance $2(k-1)$ and cardinality $|\C^k|=q^{2(n-k)}+q^{2(n-(k+(k-1)))}+\ldots
+q^{2(n-(\sum_{i=3}^{k}i))}+\Gauss{n-(\sum_{i=3}^{k}i)}{2} $.
\end{theorem}
\begin{proof}
%
First observe that, for all identifying vectors except $v_{00}^{k}$, the additional line of dots of the corresponding Ferrers diagrams does not increase the cardinality compared to the previous recursion step, due to Lemma~\ref{lm:gen_k-2}.
The only  identifying vector that contributes additional words to $\C^{k}$ is $v_{00}^k$, and thus $|\C^k|=|\C^{k-1}|+q^{2(n-k)}$ for any $k\geq 4$. Inductively, the cardinality formula follows, together with the cardinality formula for $k=3$ from Construction 0.

Next we prove that the minimum distance of $\C^k$ is $2(k-1)$.
Let $X,Y\in \C^k$, $X\neq Y$. If $v(X)=v(Y)$ then by Lemma~\ref{lem:dist_lift} $d_S(X,Y)\geq 2(k-1)$.
\\%
Now we assume that $v(X)\neq v(Y)$. Note, that according to the definition of the identifying vectors, $d_S(X,Y)\geq d_H(v(X,v(Y))=2(k-1)$  for $(X,Y)\in \C_*^k\times \C_i^k$, $0\leq i\leq 3$, for $(X,Y)\in \C_0^k\times \C_0^k$, and for  $(X,Y)\in \C_i^k\times\C_j^k$, $i\neq j$. \\
Now let $X,Y\in \C_i^k$, for some $1\leq i\leq 3$.
\begin{itemize}
  \item If the suffixes of $v( X)$ and $v( Y)$ of length $n-s$ belong to the same class $P_t$, then $d_H(v(X), v(Y))=4$ and $d_R(B_X,B_Y)= k-3$, for the common pending blocks  submatrices $B_X, B_Y$ of $\cF(X), \cF(Y)$. Then by Theorem~\ref{thm4}, $d_S(X, Y)\geq 4+2(k-3)=2(k-1)$.
  \item If the suffixes of $v( X)$ and $v( Y)$ of length $n-s$  belong to different classes, say $P_{t_1},P_{t_2}$ respectively, then $d_H(v(X), v(Y)) \geq 2$ and $d_R(B_X,B_Y)= k-2$, for the common pending blocks  submatrices $B_X, B_Y$ of $\cF(X), \cF(Y)$. Then by Theorem~\ref{thm4}, $d_S(X, Y)\geq 2+2(k-2)=2(k-1)$.
\end{itemize}
Hence, for any $X,Y \in \C^{k}$ it holds that $d_{S}(X,Y)\geq 2(k-1)$.
\end{proof}

\begin{corollary}
Let $n\geq s+2+k$ and
$q^2+q+1\geq \ell$, where $s=\sum_{i=3}^{k}i$ and $\ell=n-s$ for odd $n-s$ (or $\ell= n-s-1$ for even $n-s$). Then
$$A_{q}(n, 2(k-1), k) \geq \sum_{j=3}^k q^{2(n-\sum_{i=j}^k i)}+\Gauss{n-(\sum_{i=3}^{k}i)}{2}.$$
\end{corollary}

\begin{example}
Let $k=4$, $d=6$, $n=13$, and $q=2$.  The code $\C^4$ obtained by Construction II has cardinality $2^{18}+2^{12}+\Gauss{6}{2}=2^{18}+4747$ (the largest previously known code is of cardinality $2^{18}+4357$~\cite{EtSi09}).
\end{example}

\begin{example}
Let $k=5$, $d=8$, $n=19$, and $q=2$.  The code $\C^5$ obtained by Construction II has cardinality $2^{28}+2^{20}+2^{14}+\Gauss{7}{2}=2^{28}+1067627$ (the largest previously known code is of cardinality $2^{28}+1052778$~\cite{EtSi09}).
We illustrate now the construction.

First, we partition the set of suffixes $y\in F_2^7$ of weight $2$ into $7$ classes, $P_1,\ldots,P_7$ of size $3$ each. The identifying vectors of the code are partitioned as follows:
\begin{align*}
v_{00}^5=&(11111||0000||000||0000000),\\
\cA_0^5=\{
&(00001||1111||000||0000000),\\
&(00010||0001||111||0000000)\}\\
\cA_1^5=\{&(00100||0010||001||y):y\in P_1\}\\
\cA_2^5=\{&(01000||0100||010||y):y\in \{P_2,P_3\}\}\\
\cA_3^5=\{&(10000||1000||100||y):y\in \{P_4,P_5,P_6,P_7\}\}
\end{align*}

To demonstrate the idea of the construction we will consider only the set $\cA_2^5$. All the codewords corresponding to $\cA_2^5$ have the following common pending block $B$:
 \begin{footnotesize}
  \[\begin{array}{cccccccc}
  \bullet & \bullet & \bullet & \bullet  & \bullet & \bullet& \bullet  &  \bullet\\
   &&&& \bullet & \bullet & \bullet & \bullet \\
   &&& &&&& \bullet
        \end{array}
\]
\end{footnotesize}
If the suffix $y\in P_2$, or $y\in P_3$ then to distinguish between these two classes we assign the following values to $B$, respectively:
 \begin{footnotesize}
  \[\begin{array}{cccccccc}
  1 & \bullet & \bullet & \bullet  & \bullet & \bullet& \bullet  &  \bullet\\
   &&&& 1 & \bullet & \bullet & \bullet \\
   &&& &&&& 1
        \end{array},\textmd{ or }
        \begin{array}{cccccccc}
  0 & \bullet & \bullet & \bullet  & \bullet & \bullet& \bullet  &  \bullet\\
   &&&& 0 & \bullet & \bullet & \bullet \\
   &&& &&&& 0
        \end{array}
\]
\end{footnotesize}
For all $3$ identifying vectors with the suffixes $y$ from $P_i$, $i=2,3$, we
construct a FDMRD code of distance $2$ for the remaining dots of $B$ (here, $a=0$ or $a=1$):
\[\begin{array}{cccccccc}
  $a$ & 0 & 0 & 0  & 0 & 0& 0  & 0\\
   &&&& $a$ &0 & 0 & 0 \\
   &&& &&&& $a$
        \end{array},
        \begin{array}{cccccccc}
   $a$ & 1 & 0 & 0  & 0 & 0& 0  & 0\\
   &&&& $a$ &1 & 0 & 0 \\
   &&& &&&& $a$
        \end{array},
        \]
        \[
        \begin{array}{cccccccc}
   $a$ & 0 & 1 & 0  & 0 & 0& 0  & 0\\
   &&&& $a$ &0 & 1 & 0 \\
   &&& &&&& $a$
        \end{array},
        \begin{array}{cccccccc}
   $a$ & 0 & 0 & 1  & 0 & 0& 0  & 0\\
   &&&& $a$ &0 & 0 & 1 \\
   &&& &&&& $a$
        \end{array},
        \]
\[\begin{array}{cccccccc}
  $a$ & 1 & 1 & 0  & 0 & 0& 0  & 0\\
   &&&& $a$ &1 & 1 & 0 \\
   &&& &&&& $a$
        \end{array},
        \begin{array}{cccccccc}
   $a$ & 1 & 0 & 1  & 0 & 0& 0  & 0\\
   &&&& $a$ &1 & 0 & 1 \\
   &&& &&&& $a$
        \end{array},
        \]
        \[
        \begin{array}{cccccccc}
   $a$ & 0 & 1 & 1  & 0 & 0& 0  & 0\\
   &&&& $a$ &0 & 1 & 1 \\
   &&& &&&& $a$
        \end{array},
        \begin{array}{cccccccc}
   $a$ & 1 & 1 & 1  & 0 & 0& 0  & 0\\
   &&&& $a$ &1 & 1 & 1 \\
   &&& &&&& $a$
        \end{array}.
        \]
      Since $P_{i}$ contains only three elements, we only need to use three of the above tableaux.
 We proceed analogously for the pending blocks of $\cA_{1}^{5}, \cA_{3}^{5}$. Then we fill the Ferrers diagrams corresponding to the last $7$ columns of the identifying vectors with an FDMRD code of minimum rank distance $4$ and lift these elements. Moreover, we add the lifted MRD code corresponding to $v_{00}^{5}$, which has cardinality $2^{28}$. The number of codewords which corresponds to the set $\cA_0^5$  is $2^{20}+2^{14}$. The number of codewords that correspond to $\cA_1^5\cup\cA_2^5\cup\cA_3^5$ is $\Gauss{7}{2}$.

\end{example}

\section{Conclusion}

In this work we presented new constructions (based on the ideas of \cite{EtSi09,EtSi12,TrRo10}) of constant dimension codes in $\Gr$ with minimum subspace distance $4$ or $2k-2$, respectively. These constructions give rise to lower bounds on the cardinality of such codes, which are tighter than other known bounds for general parameters. On the other hand there exist some parameter sets where we know better constructions, hence these bounds are not tight in general.
Then again, we show some examples where our constructions come up with the largest codes known so far for the given parameters.

For future work one can try to apply the ideas of this paper to constant dimension codes with other minimum subspace distance than $4$ or $2k-2$ to come up with better codes than known so far.


\section*{Acknowledgment}
The second author war partially supported by Swiss National Science Foundation grant no. 138080 and Forschungs-kredit of the University of Zurich, grant no. 57104103.


\begin{thebibliography}{1}


    \bibitem{BoGa09}
M. Bossert and  E. M. Gabidulin,
\emph{One family of algebraic codes for network coding}, In ISIT 09,  pp. 2863 - 2866, June 2009.


\bibitem{EtSi09}
T. Etzion and N. Silberstein, ''Error-correcting codes in
projective space via rank-metric codes and Ferrers diagrams'',
\emph{IEEE Trans.\ Inform. Theory}, vol. 55, no.7, pp. 2909--2919,
July 2009.


\bibitem{EtSi12}
T. Etzion and N. Silberstein, '' Codes and designs related to lifted MRD codes '',
to appear \emph{IEEE Trans.\ Inform. Theory}.


\bibitem{EV08}
T. Etzion and A. Vardy, ``Error-correcting codes in projective
space'', \emph{IEEE Trans.\ Inform. Theory}, vol.\,57, no.\,2,
pp.\,1165--1173, February 2011.

\bibitem{Gab85}
E. M. Gabidulin, ``Theory of codes with maximum rank distance,''
\emph{Problems of Information Transmission}, vol.~21, pp. 1-12,
July 1985.

 \bibitem{ga}
 E. M. Gabidulin, N. I. Pilipchuk, ``Multicomponent network coding'', in proc. of \emph{Workshop on coding and cryptography}, pp. 443-452, 2011.


   \bibitem{GaYa10}
M. Gadouleau and Z. Yan, ``Constant-rank codes and their connection to constant-dimension codes,''
 \emph{IEEE Trans. Inform. Theory},  vol.~56, no. 7, pp. 3207--3216, July 2010.

\bibitem{KK}
R.\ Koetter and F.\,R.\ Kschischang, ``Coding for errors and
erasures~in~random network coding,'' \emph{IEEE Trans.\ Inform.
Theory}, vol.~IT-54, pp. 3579-3591, August 2008.


  \bibitem{KoKu08}
A.\,Kohnert and S.\,Kurz, ``Construction of large
constant-dimension~codes with a prescribed minimum distance,''
\emph{Lecture Notes in Computer Science}, vol.\,5393,
pp.\,31--42, December 2008.


\bibitem{vLWi92}
J. H. van Lint and R. M. Wilson,
    {\em A course in Combinatorics},
    Cambridge University Press, 2001 (second edition).


\bibitem{MGR08}
F. Manganiello, E. Gorla, and J. Rosenthal, ``Spread codes and spread decoding in network
coding'', in ISIT 08, pp. 881--885,
July 2008.

\bibitem{Rot91}
R. M. Roth, ``Maximum-rank array codes and their application to
crisscross error correction,'' \emph{IEEE Trans.\ Inform. Theory},
vol.~IT-37, pp. 328-336, March 1991.



\bibitem{Ska10}
V. Skachek, ``Recursive code construction for random networks,''
 \emph{IEEE Trans. Inform. Theory},  vol.~56, no. 3, pp. 1378--1382, March 2010.


\bibitem{SKK}
D. Silva, F.\,R.\ Kschischang, and R.\ Koetter, ``A Rank-metric
approach to error control in random network coding,'' \emph{IEEE
Trans.\ Inform. Theory},  vol.~IT-54, pp. 3951-3967, September
2008.



\bibitem{tr}
A.-L. Trautmann, ``A lower bound for constant dimension codes from multi-component lifted MRD codes'', arXiv:1301.1918


\bibitem{TrRo10}
A.-L. Trautmann  and J. Rosenthal, ``New improvements on the
echelon-Ferrers construction'', in proc. of \emph{Int. Symp. on
Math. Theory of Networks and Systems}, pp. 405--408, July 2010.



 \end{thebibliography}

\end{document}